\documentclass[11pt, amsmath, amssymb, amsfonts, nofootinbib, altaffilletter, showkeys, showmsc]{revtex4}
\usepackage[T1]{fontenc}
\usepackage[english]{babel}
\usepackage[paperwidth=210mm,paperheight=297mm,centering,hmargin=2cm,vmargin=2.5cm]{geometry}
\usepackage{amsthm}
\usepackage{hyperref}
\usepackage[dvipsnames]{color}

\newtheorem{thm}{Theorem}[section]
\theoremstyle{remark}

\theoremstyle{definition}
\newtheorem{defn}[thm]{Definition}

\newtheorem{prop}{Proposition}
\newtheorem{cor}{Corollary}
\newtheorem{rmk}{Remark}
\usepackage[all]{xy}
\usepackage{comment}
\newcommand{\reals}{\mathbb{R}}
\newcommand{\complex}{\mathbb{C}}
\newcommand{\integers}{\mathbb{Z}}

\newcommand{\gl}[1]{\text{GL}^+(#1)}
\newcommand{\glt}[1]{\widetilde{\text{GL}^+}(#1)}
\newcommand{\so}[2]{\text{SO}_{#2}(#1)}
\newcommand{\spin}[2]{\text{Spin}_{#2}(#1)}
\newcommand{\GL}[1]{\text{GL}_{#1}^{+}}
\newcommand{\GLt}[1]{\widetilde{\text{GL}_{#1}^{+}}}
\newcommand{\SO}[1]{\text{SO}_{#1}}
\newcommand{\Spin}[1]{\text{Spin}_{#1}}
\newcommand{\OO}[1]{\text{O}_{#1}}
\newcommand{\GLL}[1]{\text{GL}_{#1}}
\newcommand{\HH}{\mathcal{H}}
\newcommand{\zedtwo}{\integers/2\integers}

\newcommand{\clifford}{\complex \text{l}}
\newcommand{\vol}[2]{\text{vol}_{#2}(#1)}
\newcommand{\diff}[1]{\mathop{\text{Diff}^+(#1)}}
\newcommand{\diffo}[1]{\mathop{\text{Diff}_0^+(#1)}}
\newcommand{\diffspin}[2]{\text{Diff}^+_{#2}(#1)}
\newcommand{\diffspint}[2]{\widetilde{\text{Diff}}^+_{#2}(#1)}
\newcommand{\met}[1]{\text{Met}(#1)}
\begin{document}

\title{Dirac operator on spinors and diffeomorphisms}
\date{\today}
\author{Ludwik D\k{a}browski}
\email{dabrow@sissa.it}
\author{Giacomo Dossena}
\email{dossena@sissa.it}
\affiliation{SISSA - Via Bonomea 265 - 34136, Trieste - Italy}

\begin{abstract}

The issue of general covariance of spinors and related objects is reconsidered. Given an oriented  manifold $M$, to each spin structure $\sigma$ and Riemannian metric $g$ there is associated a space $S_{\sigma, g}$ of spinor fields on $M$ and a Hilbert space $\HH_{\sigma, g}= L^2(S_{\sigma, g},\vol{M}{g})$ of $L^2$-spinors of $S_{\sigma, g}$. The group $\diff{M}$ of orientation-preserving diffeomorphisms of $M$ acts both on $g$ (by pullback) and on $[\sigma]$ (by a suitably defined pullback $f^*\sigma$). Any $f\in \diff{M}$ lifts in exactly two ways to a unitary operator $U$ from $\HH_{\sigma, g} $ to $\HH_{f^*\sigma,f^*g}$. The canonically defined Dirac operator is shown to be equivariant with respect to the action of $U$, so in particular its spectrum is invariant under the diffeomorphisms.  
\end{abstract}


\maketitle


\section{Introduction}
In this note the issue of general covariance of spinor fields (for brevity: spinors) and related objects is reconsidered. This question has in fact at least two aspects regarding the transformation rules with respect to two different (though intrinsically related) operations: a change of coordinate system, and a diffeomorphism. In physics literature one can sometimes find statements like ``spinors transform as `spinors' with respect to the former and as scalars with respect to the latter''. While these statements can in a certain sense be justified, they are meaningful only after introducing certain mathematical structures and determining their transformation properties, as we shall explain in the next sections.

Even though in principle one usually works with vector (bilinear) or tensor (multilinear) combinations of spinors, or even with invariants (scalars) like the Lagrangian, a transformation rule of spinor fields is really needed if one wants to treat them as independent variables (e.g. with respect to some variational principle).

However a subtlety with spinors, as compared to tensors, is that one needs to work with particular double covers of the groups we are accustomed to in the case of tensors. The global mathematical constructs needed for this task have been developed in the second half of the last century \cite{Hae56}, \cite{BorHir59}, \cite{Mil63}, via the notion of spin structure. The notion of a spin structure $\sigma$ is topological in nature, but for our purposes it is here considered as an auxiliary tool to the definition of spinors and, as such, it requires a Riemannian metric $g$ to be specified on a given (oriented) smooth manifold $M$. More precisely one needs a prolongation of the principal $\SO{n}$ bundle $\so{M}{g}$ of oriented $g$-orthonormal frames to the group $\Spin{n}$. Then there is the associated space of smooth spinor fields $S_{\sigma, g}$ and the Hilbert space $\HH_{\sigma,g}= L^2(S_{\sigma,g},\vol{M}{g})$. It should be stressed that the notion of spin structure is not only sufficient, but in fact necessary for the consistency of the definition of spinor fields.

The question of the change of coordinates is then translated to the transformation rules under the change of a local orthonormal frame and corresponding change of the local spinor frame. We shall understand such a change as an automorphism of the tangent bundle, the related automorphism of the bundle of frames, and its lift to a spin structure. It should be mentioned that a large automorphism (i.e. not belonging to the connected component of the group of automorphisms) may require however a change of the spin structure $\sigma$. As far as a diffeomorphism is concerned, it is its derivative (tangent map) that plays the role of the automorphism in question.

In all these cases we shall be able to give a transformation rule of spinor fields, i.e. define a new spinor field. This new spinor field, unless the automorphism respects the metric (so the diffeomorphism is an isometry), will in general be a spinor field associated to a different metric, namely the pull back of the original metric. More precisely, we are able to give the components of the new spinor field with respect to the transformed frame (or more precisely transformed spinor frame). We should stress at this point that remaining solely in the aforementioned framework does not permit to describe the components of a given one and the same spinor field with respect to two linear frames which are orthonormal with respect to two different metrics\footnote{Actually conformally related frames can still be treated.}. This becomes possible however if the theory allows spinors with an infinite number of components (which carry a faithful representation of a double covering of the oriented general linear group). Such an extension is not usually appreciated (see however \cite{Sij75}, \cite{NeeSij79}).

As far as the group $\diff{M}$ of orientation preserving diffeomorphisms of $M$ is concerned, it acts both on $g$ (by a pull-back) and on $[\sigma]$ (by a suitably defined pull-back $f^*\sigma$). In this note we  show that any $f\in \diff{M}$ lifts (in exactly two ways) to a unitary operator from $\HH_{\sigma, g} $ to $\HH_{f^*\sigma, f^*g}$. This provides a kind of a unitary implementation on $\HH_{\sigma, g}$ of the action of a certain double covering $\diffspint{M}{\sigma}$ of the subgroup $\diffspin{M}{\sigma}$ of $\diff{M}$ preserving the spin structure $\sigma$, so in particular of the connected component of $\diff{M}$. Moreover we prove that the canonically defined Dirac operator is shown to be equivariant with respect to these actions, so in particular its spectrum is invariant under the diffeomorphisms.

In this paper we work with smooth (oriented) manifolds and use component-free notation, the usual spinor or vector indices can be easily inserted. We take the components of spinors are usual numbers, but our discussion applies in the anticommuting (Grassmann) case as well.

\section{Spinors}\label{sec:spinors}

We start by describing some algebraic structures behind spinors of a finite dimensional Euclidean space. They will be used to describe the structures on typical fibers of various bundles we shall encounter on manifolds (alternatively think of what happens at a point of a manifold).

Recall that a usual vector or tensor of $\reals^n$ of type $R$, where $R$ is some representation of $\GL{n}$ in $\reals^k$, can be viewed as a map from the space $F$ of oriented linear frames in $\reals^n$ to $\reals^k$, which is $\GL{n}$-equivariant i.e. it intertwines the canonical action of $\GL{n}$ on $F$ with $R$. Equivalently, given any (positive or negative) definite bilinear form (metric) $g$ on $\reals^n$, one can work with the space $F_g$ of oriented $g$-orthonormal frames in $\reals^n$, that carries a natural action of $\SO{n}$. Then, we can regard a vector (or a tensor) as a map from $F_g$ to $\reals^k$, which is equivariant under (restriction of) the representation $R$ to  $\SO{n}$.

For spinors one usually uses the (nontrivial) double cover $\rho\colon \Spin{n} \to \SO{n}$, and a free orbit $\tilde F_g$ of  $\Spin{n}$ (called space of `spinor frames' of $\mathbb{R}^n$) together with a 2:1 identification map $\eta \colon\tilde F_g \to F_g $, such that $\eta (\tilde e h) = \eta  (\tilde e) \rho (h)$, where $\tilde e\in\tilde F_g$ and $h\in \Spin{n}$. Given a representation $R\colon\Spin{n}\to \text{GL}(k,\complex)$ of $\Spin{n}$ we shall view a  $R$-spinor of $\mathbb{R}^n$ as an $R$-equivariant map $\psi$ from $\tilde F_g$ to $\complex^k$. There is an obvious $\complex$-linear structure on the space of $R$-spinors.

Obviously the interesting case here is when $R$ is not a tensor representation, i.e. does not descend to a representation of $\SO{n}$. This is the case e.g. for $R=\mu$, where $\mu$ is the restriction to $\Spin{n}$ of the fundamental (also called spin) representation of the Clifford algebra $\clifford_n$. The carrier complex space of $\mu$ has dimension $k=2^{n/2}$ for even $n$ and $k=2^{(n-1)/2}$ for odd $n$.

\begin{rmk}
Since $\Spin{n}$ is compact, by averaging over it we can consider any of its representations (hence also $\mu$) as being unitary with respect to a suitable hermitian inner product. In the case of the standard metric and the representation $\mu$, this is just the standard inner product on $\complex^{k}$.
\end{rmk}

Hereafter we fix the spin representation $R$ to be $\mu\colon\Spin{n}\to U(k)$, $k=2^{[n/2]}$, and consider  $\mu$-spinors $\psi:\tilde F_g\to  \complex^{k}$, i.e. $\psi(uh)=\mu(h^{-1})\psi(u), \forall u\in \tilde F_g, h\in \Spin{n}$. The inner product is given in terms of the standard inner product in $\complex^{k}$, as $({\psi,\phi}):=(\psi(u)\mid\phi(u))$ (the right hand side is independent of $u$).

In order to liberate the setting from the dependence on the metric a natural temptation would be to use the unique nontrivial for $n\geq 2$ (and universal for $n\geq 3$) double cover $\tau\colon\GLt{n}\to \GL{n}$. It extends the double cover $\rho$ and is a central extension of $\GL{n}$ by $\zedtwo$. Unfortunately  $\GLt{n}$ is not usually used as a `structure' group for spinors, for the reason that it is not a matrix group, i.e. it has only infinite-dimensional faithful representations, while geometric objects are usually assumed to have finite number of components. Instead, every finite-dimensional representation of $\GLt{n}$ descends to a (tensor) representation of $\GL{n}$, at least for $n\geq 3$ (see Lemma 5.23 in \cite{LawMic89}). Thus we have to stick to the subgroup $\Spin{n}$ and so the space of spinors will be always labelled by a metric. Concretely, a spinor labelled by a metric $g$ will be a $\mu$-equivariant map from the orbit $\tilde{F_g}:=\eta^{-1} (F_g)\subset \tilde{F}$ of $\Spin{n}$ to $\mathbb{R}^k$. 
We shall however employ $\GLt{n}$, as well as its  free orbit space $\tilde{F}$ together  with a 2:1 covering map $\eta :\tilde F \to F $, that intertwines the relative actions, 
in order to define the transformation of spinors under an oriented automorphism $\beta$ of $\mathbb{R}^n$. More precisely we can and shall lift $\beta$ to an automorphism $\tilde\beta$ of $\tilde{F}$ and define the transformed spinor as 
$$\psi' = \psi \circ  \tilde\beta_g,$$
where $\tilde\beta_g$ is the restriction of $\tilde\beta$ to $\widetilde{F_{\beta^*g}}$.
The domain of $\psi'$, understood as an equivariant map, is $\widetilde{F_{\beta^*g}}$.
Clearly the new spinor $\psi'$ is labelled by the pullback metric $g' = \beta^*\, g$. 
Note that the components of $\psi'$ with respect to the spinor basis $\tilde{e}'$ are equal to the components of $\psi$ with respect to the spinor basis $\tilde e=\tilde\beta(\tilde e')$, i.e. $\psi'(\tilde e')=(\psi\circ \tilde\beta)(\tilde e')=\psi(\tilde e)$. Moreover, since for any $\beta$ there are precisely two lifts $\tilde\beta$ (which differ just by a sign) we get actually a double covering $\widetilde{Aut^+(\mathbb{R}^n)}\cong \GLt{n}$ of the group $Aut^+(\mathbb{R}^n) \cong\text{GL}^+_n$ that acts on spinors.

In the next sections we shall globalize the structures described so far.

\section{Spin structures without metric}
In the literature the notion of a spin structure is usually formulated for a Riemannian manifold $(M,g)$ in terms of a principal $\Spin{n}$-bundle over $M$ double covering the bundle $\so{M}{g}$ of oriented $g$-orthonormal frames of $M$. Since in the following we shall vary the metric $g$, we use another though topologically equivalent definition \cite{DabPer86}, \cite{Dab88}, that is better suited for that purpose. First, let us fix notation. Given an oriented smooth $n$-manifold $M$ we denote by $\gl{M}\to M$ (or simply by $\gl{M}$) the principal $\GL{n}$-bundle of oriented frames of $M$.
\begin{defn}\label{milnor}
A spin structure on $M$ is a $\GLt{n}$-prolongation of $\gl{M}$, that is a pair $\sigma=(\glt{M},\eta)$ where $\glt{M}$ is a principal $\GLt{n}$-bundle over $M$ and the map $\eta\colon\glt{M}\to\gl{M}$ makes the following diagram commute:
\begin{displaymath}
	\xymatrix{	\glt{M}\times\GLt{n} \ar[r] \ar[dd]^{\eta\times\tau} &\glt{M} \ar[dd]^\eta \ar[dr]&\\
					&&M\\
					\gl{M}\times\GL{n} \ar[r] &\gl{M} \ar[ur]&}
\end{displaymath}
Two spin structures $(\glt{M}_1, \eta_1)$ and $(\glt{M}_2, \eta_2)$ are equivalent if there is a principal $\GLt{n}$-morphism $m\colon \glt{M}_1\to\glt{M}_2$ such that the following diagram commutes:
\begin{displaymath}
	\xymatrix{	\glt{M}_1 \ar[rr]^{m} \ar[dr]^{\eta_1} & & \glt{M}_2 \ar[dl]_{\eta_2}\\
					&\gl{M}&}
\end{displaymath}
\qed
\end{defn}
Importantly we note that given a $\GLt{n}$-prolongation $(\glt{M},\eta)$ and introducing a metric $g$ on $M$, we immediately obtain the usual definition of a spin structure as a $\Spin{n}$-prolongation of $\so{M}{g}$. Namely it suffices to consider the subbundle $\eta^{-1}(\so{M}{g})\subset \glt{M}$ with prolongation map $\eta_g\equiv\eta\restriction{\eta^{-1}(\so{M}{g})}$, the $\Spin{n}$-action being obtained by restriction of the $\GLt{n}$-action (see \cite{DabPer86}).

Conversely, any usual spin structure can be extended to a $\GLt{n}$-prolongation and defines a (metric independent) spin structure in our sense. It is also easy to see that both these processes preserve the notions of equivalence of spin structures. 

We refer to \cite{Swi93} for a nice survey about the two different definitions and their equivalence at the topological level. Hereafter we stick mainly to our notion of (metric independent) spin structure. An oriented manifold $M$ is called spin if $w_2(M)=0$ and sometimes we shall understand by this term a pair $(M,\sigma)$ with a given spin structure $\sigma$ on $M$.

\begin{rmk}
A word of caution must be given: in some texts by a spin structure an 
equivalence class is understood of spin structure in our sense, i.e. $\GLt{n}$- or $\Spin{n}$-prolongations (see e.g. \cite{Mil65} and \cite{HusJoaJur08} p.61). Moreover, sometimes it is not clearly stated if a prolongation is meant or rather an equivalence class of prolongations, though this may be grasped from context. Clearly this is crucial for the issue of proper parametrization, e.g. it is the set of equivalence classes of spin structures on $M$ which is known to be in bijective correspondence with $H^1(M,\zedtwo)$ see e.g. \cite{Dab88}. Note also the difference with the case of reductions of the structure group -- a reduction of a principal $G$-bundle $P$ to some subgroup $G'\subset G$ is a principal $G'$-subbundle $P'\subset P$. Two different reductions encode different information with respect to the inclusion into $P$, even though they might be equivalent as reductions (the equivalence is defined analogously as for prolongations). For instance two different $\OO{n}$-reductions of the principal $\GLL{n}$-bundle of frames over a manifold $M$ correspond to different Riemannian metrics on $M$, even though any two such reductions are equivalent.
\end{rmk}

\section{Spinor fields}\label{sec:spinorfields}
Let  $(M,g,\sigma)$ be a Riemannian spin manifold. Let $R\colon\Spin{n}\to \text{GL}(k,\complex)$ be a (fixed) representation of $\Spin{n}$. It is customary to call $R$-spinor field
on $M$ an $R$-equivariant map $\psi\colon \spin{M}{g}\to \complex^k$, where $\spin{M}{g}$ is the total space of the spin structure $\sigma$ (here by spin structure we temporarily mean a $\Spin{n}$-prolongation of $\so{M}{g}$, i.e. we need the metric $g$). By $R$-equivariance we mean that $\psi(ug)=R(g^{-1})\psi(u)$ for $u\in \spin{M}{g}$ and $g\in\Spin{n}$. As in Section \ref{sec:spinors} we are interested in those $R$ that are not  tensor representations, e.g. in the unitary Dirac representation $\mu$ of $\Spin{n}$ in the complex space of dimension $k=2^{[n/2]}$. We denote by $S_{\sigma,g}$ the space of $\mu$-spinor fields, often named Dirac spinor fields, for the spin structure $\sigma $ and the metric $g$. There is an obvious $\complex$-linear structure on $S_{\sigma,g}$ induced by pointwise operations. Note that  for different metrics we have a priori different spaces $S_{\sigma,g}$ (see \cite{Swi93} for a geometric description of a configuration space for both spinors and metrics).

An inner product on $S_{\sigma,g}$ can be defined as follows: take a cover $\{(U_\alpha, h_\alpha)\}_{\alpha\in A}$ of $M$ which trivializes $\spin{M}{g}$. Given $\psi,\phi\in S_{\sigma,g}$ consider the global function $a_{\psi,\phi}\colon M\to\complex$ defined locally by $$a_{\psi,\phi}(x):=(\psi(u_x)\mid\phi(u_x))$$ where $u\colon U_\alpha\to \spin{M}{g}(U_\alpha)$ is any local section of $\spin{M}{g}$ (in writing $u$ we omit the dependency on the index $\alpha$ to simplify notation). Consider the global (yet locally defined, $\alpha$-dependency omitted) $n$-form $e^1\wedge \cdots\wedge e^n$ where $e^j$ is the $g$-dual of $e_j:=\eta_g(u_j)$. Finally put: $$\langle\psi\mid\phi\rangle_{\sigma, g}:=\int_M a_{\psi,\phi} \;e^1\wedge\cdots\wedge e^n\quad.$$
It is easy to see that the above definition does not depend on the trivialization. Note that $e^1\wedge \cdots\wedge e^n=\vol{M}{g}$, the $g$-volume form of $M$.

We then make the following definition.

\begin{defn}
The Hilbert space of spinors $\HH_{\sigma,g}$ for a given spin structure $\sigma$ and metric $g$ is the $L^2$-completion of the inner product space $(S_{\sigma,g},\langle~\mid~\rangle_{\sigma, g})$.
\end{defn}

It is natural to investigate what happens to $\HH_{\sigma, g}$ under a change of spin structure. For equivalent spin structures the answer is given by the next proposition.

\begin{prop}
If we choose an equivalent $\GLt{n}$-prolongation $\sigma'=(\glt{M}',\eta')$, the principal $\GLt{n}$-isomorphism $m\colon \glt{M}'\to\glt{M}$ induces a unitary operator $U\colon\HH_{\sigma,g}\to\HH_{\sigma',g}$ given by $$U\psi=\psi\circ m_g$$ where $m_g=m\restriction{\eta'^{-1}(\so{M}{g})}$.
\end{prop}
\begin{proof}
The operator $U$ is clearly linear. It is invertible with inverse given by $U^{-1}\psi=\psi\circ (m_g)^{-1}$. To prove unitarity let us put $\psi'=\psi\circ m_g$. It is now enough to observe that $a_{\psi',\phi'}=a_{\psi,\phi}$. From this we obtain $\langle U\psi\mid U\phi\rangle=\int_M a_{\psi',\phi'}\vol{M}{g}=\int_Ma_{\psi,\phi}\vol{M}{g}=\langle \psi\mid\phi\rangle$.
\end{proof}
\begin{rmk}\label{rmk:twolifts}
Given two equivalent $\GLt{n}$-prolongations of $\GL{M}$, there are exactly two distinct principal isomorphisms between the two prolongations (this is a consequence of the morphism $\tau$ being a central extension of $\GL{n}$ by $\zedtwo$). It follows that there is another unitary operator $U^-\colon\HH_{\sigma,g}\to\HH_{\sigma',g}$, given by $U\psi=\psi\circ m_g^-$ where $m_g^-u=(mu)(-1)$ with $\{\pm1\}=\ker\tau\subset Z(\GLt{n})$. Clearly, once the existence of an isomorphism $\HH_{\sigma,g}\to\HH_{\sigma',g}$ has been established, any other isomorphism can be obtained by composing with a suitable automorphism of $\HH_{\sigma,g}$. However, the operators $U$ and $U^-$ are the only two arising from principal morphisms as indicated above.
\end{rmk}

\section{Diffeomorphisms and spin structures}

We now study the interplay between orientation-preserving diffeomorphisms of $M$ and spin structures on $M$. Given a spin manifold $(M,\sigma)$ where $\sigma=(\glt{M},\eta)$, let us choose an orientation-preserving diffeomorphism $f\colon M\to M$ and consider the natural lift of $f$ to $\gl{M}$ given by applying the tangent map of $f$ to each element of each frame $e\in\gl{M}$. We denote such a lift by the symbol $Tf$. The pullback bundle $Tf^*\glt{M}$, defined explicitly by $Tf^*\glt{M}=\{(e,u)\in\gl{M}\times\glt{M}\mid Tf(e)=\eta(u)\}$, together with the canonical map $Tf^*\eta\colon Tf^*\glt{M}\to\gl{M}$ given by $(Tf^*\eta)(e,u)=e$ is again a spin structure on $M$ which we call $f^*\sigma$. By construction the map $Tf\colon\gl{M}\to\gl{M}$ admits exactly two distinct lifts, given by $\varphi^\pm(e,u)=u(\pm1)$ where $\{\pm1\}=\ker\tau\subset Z(\GLt{n})$. The following diagram illustrates the situation.
\begin{displaymath}
\xymatrix{Tf^*\glt{M} \ar[r]^{\varphi^\pm} \ar[d]^{Tf^*\eta} & \glt{M} \ar[d]^\eta\\
    				\gl{M}\ar[r]^{Tf} \ar[d] & \gl{M} \ar[d]\\
				M \ar[r]^{f} & M}
\end{displaymath}
Recall that the set $\Sigma_M$ of equivalence classes of spin structures on $M$ is naturally an affine space over the $\zedtwo$-vector space $H^1(M;\zedtwo)$. The assignment $\diff{M}\times \Sigma_M\to\Sigma_M$ given by $(f,[\sigma])\mapsto [f^*\sigma]$ defines an affine representation $\rho$ of $\diff{M}$ on $\Sigma_M$ (see \cite{DabPer86} for a proof). Moreover, the normal subgroup $\diffo{M}\subset\diff{M}$ of diffeomorphisms which are homotopy equivalent to the identity acts trivially on $\Sigma_M$, hence $\rho$ descends to a representation of $\Omega(M)=\diff{M}/\diffo{M}$ on $\Sigma_M$.

\section{Diffeomorphisms and spinors}
This section explores the relation between diffeomorphisms of $M$ and the system of spaces $\HH_{\sigma, g}$. Let us start with a spin structure $\sigma=(\glt{M},\eta)$ of $M$ and a metric $g$ on $M$. Given an orientation-preserving diffeomorphism $f\in\diff{M}$ we can consider the pullback metric $f^*g$ on $M$ defined by $(f^*g)(v,w)=g(Tfv,Tfw)$. The map $Tf\colon \gl{M}\to\gl{M}$ restricts to a lift $Tf_g\colon \so{M}{f^*g}\to \so{M}{g}$ by construction. The pullback spin structure $f^*\sigma$ restricts to a $\Spin{n}$-prolongation of $\so{M}{f^*g}$ by considering $(Tf^*\eta)^{-1}(\so{M}{f^*g})$ with $\Spin{n}$-action obtained by restricting the $\GLt{n}$-action on $Tf^*\GLt{M}$ to the subbundle $(Tf^*\eta)^{-1}(\so{M}{f^*g})$. There are exactly two lifts $\varphi_g^\pm\colon (Tf^*\eta)^{-1}(\so{M}{f^*g})\to \eta^{-1}(\so{M}{g})$, given by restriction of $\varphi^\pm$. The following diagram illustrates the situation.
\begin{displaymath}
\xymatrix{(Tf^*\eta)^{-1}(\so{M}{f^*g}) \ar@{^{(}->}[r] & Tf^*\glt{M} \ar[r]^{\varphi^\pm} \ar[d]^{Tf^*\eta} & \glt{M} \ar[d]^\eta & \eta^{-1}\so{M}{g}\ar@{_{(}->}[l] \\
    				\so{M}{f^*g} \ar@{^{(}->}[r] &\gl{M}\ar[r]^{Tf} \ar[d] & \gl{M} \ar[d]&\so{M}{g}\ar@{_{(}->}[l] \\
				&M \ar[r]^{f} & M}
\end{displaymath}
The next definition and proposition generalize the analysis in Section \ref{sec:spinorfields} to the case of changing the metric from $g$ to $f^*g$.
\begin{defn}
For each of the two lifts $\varphi^\pm$ of $Tf$ we define a linear operator $U_{\varphi^\pm}\colon \HH_{\sigma, g}\to\HH_{f^*\sigma, f^*g}$ by 
\begin{equation}\label{unitary}
U_{\varphi^\pm}\psi=\psi\circ \varphi_{f^*g}^\pm
\end{equation}
where $\varphi_{f^*g}^\pm=\varphi^\pm\restriction{(Tf^*\eta)^{-1}(\so{M}{f^*g})}$.
\end{defn}
\begin{prop}\label{prop:unitary}
The operators $U_{\varphi^\pm}$ defined above are unitary, that is they are invertible and satisfy $\langle U_{\varphi^\pm}\psi\mid U_{\varphi^\pm}\phi\rangle_{f^*\sigma, f^*g}=\langle \psi\mid\phi\rangle_{\sigma, g}$ for each $\psi,\phi\in\HH_{\sigma,g}$.
\end{prop}
\begin{proof}
Linearity is clear. The inverse is given by $\psi\mapsto \psi\circ (\varphi_{f^*g}^\pm)^{-1}$. For the second part: let us consider $\varphi_{f^*g}^+$, the case $\varphi_{f^*g}^-$ being analogous. Put $\psi':=\psi\circ \varphi_{f^*g}^+$, $\phi':=\phi\circ \varphi_{f^*g}^+$. An easy computation shows that $a_{\psi',\phi'}=a_{\psi,\phi}\circ f$. Now apply the formula for the invariance of integrals under pullback:
\begin{equation}
\begin{aligned}
\langle U_{\varphi^+}\psi\mid U_{\varphi^+}\phi\rangle_{f^*\sigma, f^*g}&=\int_M a_{\psi',\phi'} \;e'^1\wedge\cdots\wedge e'^n\\
&=\int_M (a_{\psi,\phi} \circ f) \;e'^1\wedge\cdots\wedge e'^n\\
&=\int_M f^*\left(a_{\psi,\phi} \;e^1\wedge\cdots\wedge e^n\right)\\
&=\int_M a_{\psi,\phi} \;e^1\wedge\cdots\wedge e^n\\
&=\langle\psi\mid\phi\rangle_{\sigma,g}
\end{aligned}
\end{equation}
where we used local sections $e'\colon U_\alpha\to \so{M}{f^*g}(U_\alpha)$ and $e:=Tfe'$.
\end{proof}
A remark similar to \ref{rmk:twolifts} holds here as well. In other words, the operators $U_{\varphi^\pm}$ are the only two unitary operators $\HH_{\sigma, g}\to\HH_{f^*\sigma, f^*g}$ which arise from some principal morphism as above.

The above results permit to introduce a certain covering of the group of diffeomorphisms. We restrict to the case of oriented diffeomorphisms preserving a given spin structure.
\begin{defn}
Let $\diffspin{M}{\sigma}$ be the subgroup of $\diff{M}$ consisting of diffeomorphisms which preserve the spin structure $\sigma=(\glt{M}, \eta)$. Define the group $\diffspint{M}{\sigma}$ to consist of all principal $\GLt{n}$-morphisms $\varphi\colon\glt{M}\to\glt{M}$ closing the diagram:
\begin{displaymath}
\xymatrix{\glt{M} \ar[r]^{\varphi} \ar[d]^{\eta} & \glt{M} \ar[d]^\eta\\
    				\gl{M}\ar[r]^{Tf} \ar[d] & \gl{M} \ar[d]\\
				M \ar[r]^{f} & M}
\end{displaymath}
where $f$ runs over $\diffspin{M}{\sigma}$, together with the multiplication given by composition of maps.
\end{defn}

It is clear that $\diffspint{M}{\sigma}$ is a double cover of $\diffspin{M}{\sigma}$ by the map $\pi_\sigma(\varphi)= f$. The corresponding operators $U_{\varphi^\pm}$ given by \eqref{unitary} implement -- in a generalized sense -- the action on spinor fields of the double cover $\diffspint{M}{\sigma}$ of oriented, spin structure preserving diffeomorphisms. This is however not an implementation in the strict sense, as we have not really an action on a fixed space of spinors but rather the target space of spinors changes according to the pull back action of $f$ on the metric.

In order to get of a genuine action one should develop further our setting. A possible way could be to consider the disjoint union $\mathcal{C}_\sigma=\amalg_g \HH_{\sigma, g}$ where $g$ runs over all Riemannian metrics on $M$. The (right) action is then given by:
\begin{equation}
\mathcal{C}_\sigma\times \diffspint{M}{\sigma}\to\mathcal{C}_\sigma,\; 
(\psi,\varphi)\mapsto \psi\cdot\varphi:=\psi\circ\varphi_{f^*g}
\end{equation}
where $\psi\in\HH_{\sigma, g}$, $f=\pi_\sigma(\varphi)$ and $\varphi_{f^*g}=\varphi\restriction{\eta^{-1}(\so{M}{f^*g})}$.
\begin{rmk}
By Proposition \ref{prop:unitary}, this action is `fiberwise unitary' in the sense that it is linear on each component $\HH_{\sigma, g}$, it is invertible and $\langle \psi\cdot\varphi \mid \chi\cdot\varphi\rangle_{\sigma, f^*g}=\langle \psi\mid\chi\rangle_{\sigma, g}$ for each $\psi,\chi\in\HH_{\sigma, g}$.
\end{rmk}
In order to speak of a unitary action of $\diffspint{M}{\sigma}$ one should put a Hilbert space structure on $\mathcal{C}_\sigma$. It would be natural to view $\mathcal{C}_\sigma$ as a direct integral of Hilbert spaces over the space $\met{M}$ of Riemannian metrics on $M$. We hope it can be made rigorous using the following facts.  First, $\met{M}$ is a positive convex cone into the vector space of smooth symmetric $(0,2)$-tensors on $M$. The latter is naturally a Fr\'echet space by equipping it with the smooth topology. The space $\met{M}$ is open in that space, hence it inherits the structure of a Fr\'echet manifold. The tangent space of $\met{M}$ at some $g\in\met{M}$ can be identified with the vector space of smooth symmetric $(0,2)$-tensors on $M$. A Riemannian metric $\mu$ can be put on $\met{M}$ which is invariant under the action of diffeomorphisms of $M$ by pullback, $g\mapsto f^*g$. Given $\varphi\in \diffspint{M}{\sigma}$, we could then define for each $\psi\in \int^\oplus_{\met{M}} \HH_{\sigma, g}\textrm{d}\mu$ the element $(\mathfrak{U}_\varphi\psi)(g)=U_\varphi\psi(\pi_\sigma(\varphi^{-1})^*g)$, where we denote by $\textrm{d}\mu$ the induced invariant measure on $\met{M}$. The assignment $\varphi\mapsto \mathfrak{U}_\varphi$ would then be a unitary action of $\diffspint{M}{\sigma}$ on the space of spinors with spin structure $\sigma$.

\section{Equivariance of the Dirac operator}

In order to define the Dirac operator one uses the lift of the covariant derivative associated to the Levi-Civita (metric preserving and torsion free) connection. Its local components with respect to an orthonormal frame $e$ are given explicitly by the Christoffel symbols
\begin{equation}\label{connection}
\Gamma^{(e(x))}_{jkl} = c_{jkl} + c_{jlk} + c_{lkj}
\end{equation}
where $c_{ijk}$ are the structure constants of the commutators (as vector fields)
$$[ e_i, e_j ] =   c_{ijk} e_k .$$

Then for a given $\sigma$ and $g$ on $M$ the Dirac operator $D$ is defined by its local components, i.e. its action in the 'gauge' $\tilde e$ on the local components $\psi \circ \tilde{e}$ of  $\psi\in S_{\sigma, g}$ as 
$$
(D\psi) (\tilde e(x)) := \sum_j \gamma_j
\left({\cal L}_{e_j(x)} + 
\frac{1}{4} \sum_{k\ell}\gamma_k\gamma_\ell\Gamma^{(e(x))}_{jkl} \right)
\psi (\tilde e(x)),
$$
where $\gamma_j$ are the anticommuting gamma matrices and $e=\eta \circ \tilde e$.
As it should, up to a unitary equivalence the Dirac operator is independent on the choice of a representation of the gamma matrices and of local orthonormal frames.

\begin{prop}\label{prop:equivariance}
The Dirac operator is equivariant, i.e.
\begin{equation}\label{eq:cov}
D' U^\pm_f = U^\pm_f D\;,
\end{equation}
where $D'$ is the Dirac operator on $\HH_{f^*\sigma,f^*g}$.
\end{prop}
Denoting 
$\psi' := U^\pm_f \psi $
we can write \eqref{eq:cov} in the equivalent form as 
\begin{equation}\label{eq:cov2}
D'\psi' = (D\psi)'
\end{equation}

\begin{proof}
It is a matter of a straightforward check that \eqref{eq:cov2} is satisfied. For that evaluate both sides on $\tilde e'(x)$ using the fact that ${\cal L}_{e'_j}(\psi\circ f)(x)= {\cal L}_{e_j}\psi(f(x))$, that the local Christoffel symbols in any orthonormal frame are given in terms of commutators of the vectors constituting the frame and the commutators of $f$-related frames are $f$-related, and the equality of local components $\psi'(\tilde e'(x))= \psi(\tilde e(f(x)))$.
\end{proof}

From the formula \eqref{eq:cov2}, which is already present (modulo a typo) in \cite{Dab88} (p.101 at the bottom), follows that the eigenvalues (point spectrum) of the Dirac operator are invariant under diffeomorphisms.
Now using also Proposition \eqref{prop:unitary} we can state a stronger result:
\begin{cor}
The spectrum of the Dirac operator is invariant under the diffeomorphisms.
\end{cor}

\section{Brief account of some other approaches}

The relations between spinors, the Dirac equation and the metric has been investigated by other authors. Here we mention only that a particular geometrically constructed ('based' or 'vertical') bundle isomorphism $m_{g_1,g_2}\colon\so{M}{g_1}\to\so{M}{g_2}$ for any two Riemannian metrics $g_1$ and $g_2$ on $M$ can be constructed (e.g. \cite{MilSta74}, \cite{BinPfe83}). A lift of this isomorphism to respective spin bundles makes possible to compare spinors and the Dirac equation for different metrics. However, the two spin structures are necessarily equivalent.

The paper \cite{BouGau92} combines ideas in \cite{BinPfe83} and \cite{DabPer86} to construct a ``metric'' Lie derivative of spinor fields. As in \cite{DabPer86}, the metrics considered are $g$ and $f^*g$. The isomorphism $m_{g,f^*g}$ is used to project the tangent map of a diffeomorphism $f\colon M\to M$ onto the same principal bundle $\so{M}{g}$; an analogous isomorphism between the two principal $\Spin{n}$-bundles associated to $g$ and $f^*g$ is used to project the lift of $f$, thus realising an automorphism of the same $\Spin{n}$-bundle over $f$. This permits to define a Lie derivative of spinor fields, which however does not induce the canonical Lie derivative on tensor fields build from spinor fields. Its geometric nature has been clarified in \cite{GodMat03}. The procedure above works however only for strictly Riemannian metrics and spin-structure-preserving diffeomorphisms.

It is worth to mention that the canonical Dirac operator on Dirac spinor fields provides a prominent example of a spectral triple, and of a noncommutative Riemannian spin manifold in the framework of noncommutative geometry of Connes \cite{Con94}. The results of this paper fit well into this scheme and can be interpreted as a unitary implementation of diffeomorphisms on spectral triples. Concerning the additional requirements (axioms) for noncommutative Riemannian spin manifolds \cite{Con96}, most of them are preserved under diffeomorphisms in a straightforward manner. Only the axiom of projectivity and absolute continuity requires a comment. Namely it is easy to check that the $C^\infty(M)$-modules of smooth spinor fields, equipped with the $C^\infty(M)$-valued hermitian form, are intertwined by the action of diffeomorphisms, i.e 
$$(a \psi)\circ \tilde f_\pm = (f^*a) (\psi \circ \tilde f_\pm).$$

\section{Final remarks}

In this paper we have further developed the approach of \cite{DabPer86} and \cite{Dab88} to give a consistent definition of the transformation rules for spinor fields under (the double cover of) diffeomorphisms and checked the covariance of the Dirac operator. This requires however, as mentioned in the introduction, the changing of the space of spinors according to the pull back action on metrics and on spin structures labelling the spaces of spinors. In particular we are able to give the components of the transformed spinor field with respect to the transformed (spinor) linear frame, orthonormal with respect to the different (pullback) metric. It should be stressed however that we cannot compare the components of a given one and the same spinor field with respect to two linear frames if they are not related by a orthonormal transformation (giving a proper scaling dimension we could treat however conformally related metrics).

Since we have not employed an isomorphism between Hilbert spaces associated to different Riemannian metrics, we can not discuss in general the behaviour of spinors under infinitesimal diffeomorphisms and the notion of the Lie derivative on spinor fields along vector fields (unless they are Killing vector fields).

Moreover, for simplicity we considered only the oriented diffeomorphisms, the orientation changing diffeomorphisms in general would require the coverings (there are two) of the full orthogonal group, known as $Pin_{\pm}$. Some parts of our results hold as well in the Lorenzian or pseudoriemannian case. The isomorphism $\varphi$ should also play an important role for a rigorous discussion of the variational aspects of the theory (under a general variation of the metric), and thus for deriving the equation of motions.

\bibliographystyle{siam}

\begin{thebibliography}{10}

\bibitem{BinPfe83}
{\sc E.~Binz and R.~Pferschy}, {\em The {D}irac operator and the change of the
  metric}, C. R. Math. Rep. Acad. Sci. Canada, 5 (1983), pp.~269--274.

\bibitem{BorHir59}
{\sc A.~Borel and F.~Hirzebruch}, {\em Characteristic classes and homogeneous
  spaces. {II}}, Amer. J. Math., 81 (1959), pp.~315--382.

\bibitem{BouGau92}
{\sc J.-P. Bourguignon and P.~Gauduchon}, {\em Spineurs, op{\'e}rateurs de
  {D}irac et variations de m{\'e}triques}, Comm. Math. Phys., 144 (1992),
  pp.~581--599.

\bibitem{Con94}
{\sc A.~Connes}, {\em Noncommutative geometry}, Academic Press Inc., San Diego,
  CA, 1994.

\bibitem{Con96}
\leavevmode\vrule height 2pt depth -1.6pt width 23pt, {\em Gravity coupled with
  matter and the foundation of non-commutative geometry}, Comm. Math. Phys.,
  182 (1996), pp.~155--176.

\bibitem{Dab88}
{\sc L.~D{\k{a}}browski}, {\em Group actions on spinors}, vol.~9 of Monographs
  and Textbooks in Physical Science. Lecture Notes, Bibliopolis, Naples, 1988.

\bibitem{DabPer86}
{\sc L.~D{\k{a}}browski and R.~Percacci}, {\em Spinors and diffeomorphisms},
  Comm. Math. Phys., 106 (1986), pp.~691--704.

\bibitem{GodMat03}
{\sc M.~Godina and P.~Matteucci}, {\em Reductive {$G$}-structures and {L}ie
  derivatives}, J. Geom. Phys., 47 (2003), pp.~66--86.

\bibitem{Hae56}
{\sc A.~Haefliger}, {\em Sur l'extension du groupe structural d'un espace
  fibr{\'e}}, C. R. Acad. Sci. Paris, 243 (1956), pp.~558--560.

\bibitem{HusJoaJur08}
{\sc D.~Husem{{\"o}}ller, M.~Joachim, B.~Jur{\v{c}}o, and M.~Schottenloher},
  {\em Basic bundle theory and {$K$}-cohomology invariants}, vol.~726 of
  Lecture Notes in Physics, Springer, Berlin, 2008.
\newblock With contributions by Siegfried Echterhoff, Stefan Fredenhagen and
  Bernhard Kr{{\"o}}tz.

\bibitem{LawMic89}
{\sc H.~B. Lawson, Jr. and M.-L. Michelsohn}, {\em Spin geometry}, vol.~38 of
  Princeton Mathematical Series, Princeton University Press, Princeton, NJ,
  1989.

\bibitem{Mil63}
{\sc J.~Milnor}, {\em Spin structures on manifolds}, Enseignement Math. (2), 9
  (1963), pp.~198--203.

\bibitem{Mil65}
{\sc J.~W. Milnor}, {\em Remarks concerning spin manifolds}, in Differential
  and {C}ombinatorial {T}opology ({A} {S}ymposium in {H}onor of {M}arston
  {M}orse), Princeton Univ. Press, Princeton, N.J., 1965, pp.~55--62.

\bibitem{MilSta74}
{\sc J.~W. Milnor and J.~D. Stasheff}, {\em Characteristic classes}, Princeton
  University Press, Princeton, N. J., 1974.
\newblock Annals of Mathematics Studies, No. 76.

\bibitem{NeeSij79}
{\sc Y.~Ne'eman and D.~{\v{S}}ija{\v{c}}ki}, {\em Unified affine gauge theory
  of gravity and strong interactions with finite and infinite {$\overline{{\rm
  GL}}(4,\,{\bf R})$} spinor fields}, Ann. Physics, 120 (1979), pp.~292--315.

\bibitem{Sij75}
{\sc D.~{\v{S}}ija{\v{c}}ki}, {\em The unitary irreducible representations of
  {$\overline{{\rm SL}}(3,\,R)$}}, J. Mathematical Phys., 16 (1975),
  pp.~298--311.

\bibitem{Swi93}
{\sc S.~T. Swift}, {\em Natural bundles. {II}. {S}pin and the diffeomorphism
  group}, J. Math. Phys., 34 (1993), pp.~3825--3840.

\end{thebibliography}

\end{document}